\documentclass[sigconf, nonacm]{acmart}

\usepackage{xspace}
\usepackage{graphicx}
\usepackage{tikz}
\usepackage{comment}
\usepackage{enumitem}
\newcommand{\CARTS}{\textsc{CARTS}\xspace}
\newcommand{\CARTSfull}{Contextual Autoregressive Rank Transcoding Steganography (\CARTS)}
\usepackage{stfloats} 

\theoremstyle{plain}
\newtheorem{theorem}{Theorem}
\newtheorem{lemma}{Lemma}
\newtheorem{prop}{Proposition}
\newtheorem{corollary}{Corollary}

\theoremstyle{definition}
\newtheorem{defn}{Definition}
\newtheorem{prob}{Problem}

\theoremstyle{remark}
\newtheorem{rmk}{Remark}
\newtheorem{ex}{Example}

\begin{document}

\title{\CARTS: Contextual Autoregressive Rank Transcoding Steganography for Full-Capacity Keyed Text Encoding}

\author{Wissam Ghantous}
\orcid{0002-2119-0232}
\affiliation{%
  \institution{University of Central Florida}
  \city{Orlando}
  \state{Florida}
  \country{USA}}
\email{wissam.ghantous@ucf.edu}

\author{Alexander V. Mantzaris}
\orcid{0002-0026-5725}
\affiliation{%
  \institution{University of Central Florida}
  \city{Orlando}
  \state{Florida}
  \country{USA}}
\email{alexander.mantzaris@ucf.edu}

\renewcommand{\shortauthors}{Ghantous et al.}

\begin{abstract}
Autoregressive language models can be used to transform a payload 
text into a stegotext of identical token length by preserving 
per-position rank information across contexts --- a methodology 
we formalize as Contextual Autoregressive Rank Transcoding 
Steganography (\CARTS). While the Calgacus construction 
\cite{norelli2025llms} demonstrated this phenomenon experimentally, 
no formal security analysis existed. This paper provides the first 
rigorous treatment of \CARTS. 
We show its exact correctness under deterministic model assumptions, 
introduce a rank-coordinate representation in which keys act as bijections on rank-vector space, 
define relevant security notions and the computational problems 
naturally associated with the construction --- context search, key collisions, message equivocation, 
and non-commutativity of the encoding maps --- 
and study the theoretical relationships between them,
including the characterization of message equivocation in terms of context search, 
and the tension between key collisions and message equivocation. 
An empirical study on Llama 3 8B 
confirms exact recovery of the original payload in all tested cases, 
finds no key collisions under random key 
generation, establishes that a hand-crafted collision is local rather 
than global, and finds no commuting key pairs --- suggesting resistance 
to the attack vectors studied. This work opens a formally grounded 
research agenda for the constructive use of language models in 
cryptography and privacy-preserving communication.
\end{abstract}

\maketitle

\section{Introduction}
\label{sec:introduction}

Recent years have seen growing interaction between artificial intelligence and cybersecurity, with researchers studying both how cryptographic techniques can be used to analyze or attack machine-learning systems, and conversely how AI systems can be used to attack or weaken existing cryptographic constructions. For example, prior work has investigated cryptographic approaches to model extraction and parameter recovery attacks \cite{canales2025extracting,carlini2025polynomial}, while other directions study the use of machine learning for cryptanalysis or automated attacks on security systems \cite{zhang2025neural,gohr2019improving,wenger2022salsa}. In contrast, the present work considers a different direction: the use of autoregressive language models themselves as the foundation for security primitives. 

Steganography aims to embed a secret payload into an innocuous-looking cover object so that a third party cannot reliably determine whether hidden communication is taking place \cite{cachin1998stego_model,hopper2002provably_secure}.
It therefore complements cryptography: encryption protects message content, while steganography targets the observability of communication itself, which matters for censorship resistance, and deniable communication.

Recent advances in neural language modeling have enabled practical \emph{linguistic} steganography, where natural language is the cover channel.
Early neural approaches used language models to generate covertext while encoding bits through constrained sampling or coding schemes \cite{ziegler-etal-2019-neural,shen-etal-2020-near}.
From a provable security perspective, security oriented work has investigated how to achieve indistinguishability relative to realistic text distributions by combining universal steganography ideas with modern generative models \cite{kaptchuk2021meteor}.
In parallel, black-box settings (where only an interface or API is available) motivate alternative constructions and threat models \cite{wu2024llm_stega}.

In a conventional linguistic steganography setting, the sender begins with an existing cover text and modifies it using a shared key to produce a stegotext carrying the hidden payload. A receiver possessing the same key can then recover the secret message.
The Calgacus construction of \cite{norelli2025llms} proposes a new way, exploiting the use of LLMs, 
to achieve steganography, producing a new text of the same length as the payload. 
Instead of editing a fixed cover text, Calgacus uses a secret prompt as the key to steer the topic and style of the generated output. 
The payload text is first converted into LLM token ranks, 
then the LLM is run under the key and forced to emit tokens whose ranks match that sequence.
Decoding repeats the same process backwards, obtaining the original message. 
This allows the user to achieve \emph{full-capacity, length-preserving text-to-text transcoding}, where a meaningful message can be hidden inside another coherent text of exactly the same token length. 

This new construction opens up a whole area of research on the constructive use of LLMs in  privacy-preserving communication, with the goal of building new primitives, and has not yet garnered too much attention. 
Only limited prior work has studied the sensitivity of the Calgacus construction to variations in the secret context, showing that small perturbations of the key can induce substantially different generated outputs \cite{Mantzaris_Ghantous_Stinebrickner_Jahangiri_Webinga_2026}.
As a consequence, many theoretical and practical questions remain unanswered. 
In particular, this construction has only been demonstrated to work experimentally, via examples, without a formal security analysis. The initial paper \cite{norelli2025llms} does not provide any concrete security definitions, nor does it identify the computational problems underlying the security properties for their construction. The goal of this paper is to fill this gap, and provide an initial concrete empirical study of the assumptions underlying the security of this new construction. 

\smallskip
\noindent
\textbf{Terminology.}
We propose the formal term \CARTSfull{} to describe this methodology.
Informally, \CARTS takes (i) a payload sentence, (ii) a private key realized as a secret context/prompt, and (iii) a fixed autoregressive language model, and outputs a new sentence of the same token length.
Decoding with the same key deterministically maps the generated
stegotext back to the original payload.
The surface text can appear benign (or policy-compliant when used over public communications channels) while the payload is recoverable only by parties holding the key and the model specification.

\smallskip
\noindent
\textbf{Contributions.}
In this paper, we (i) formalize \CARTS as a keyed rank-transcoding steganographic primitive and present the Calgacus construction of \cite{norelli2025llms} within this framework; 
(ii) introduce a set of computational problems and associated assumptions naturally motivated by the underlying LLM; 
(iii) formally define some desirable security notions for \CARTS as a steganographic protocol; 
(iv) characterize the security properties of \CARTS and the corresponding attacker tasks in terms of these computational problems (e.g., context search and key collisions)
and (v) provide an initial theoretical and experimental study of these computational problems to assess their difficulty.

\smallskip
\noindent
\textbf{Roadmap.}
In Section \ref{sec:Conceptual_Framework}, we explain the subtle conceptual difference between CARTS and the prior types of steganography. 
In Section \ref{sec:carts-protocol}, we formally introduce the general CARTS steganography framework and view Calgacus as a natural instantiation, using LLMs. 
In Section \ref{sec:security}, we begin by discussing scope and adversarial models, then introduce various desirable security properties that a steganographic scheme could have, and the corresponding computational problems, naturally motivated by the underlying LLM used. At a second stage, we study these problems theoretically, from various angles. 
In particular, we discuss matters of equivocation, key collisions, as well as the issue of commuting keys and its implications on security. 
Lastly, in Section \ref{sec:EmpiricalExplorations}, we initiate an empirical study into the relevant computational problems that arose in the body of the paper. 
In the conclusion, we provide some concrete problems to be studied next, to further solidify our theoretical and practical understanding of LLMs, as a basis for steganography.

\section{Conceptual Framework}
\label{sec:Conceptual_Framework}

Formally, in steganography, the sender begins with an innocent-looking \emph{cover} object $c$ (for example, an ordinary piece of text)
and embeds a secret payload $m$ into that cover.
The output, called a \emph{stegotext} $s$, should look like a normal cover while allowing an intended receiver to recover $m$.
Many formulations include a shared secret key $k$ that controls the embedding and extraction procedures \cite{cachin1998stego_model,hopper2002provably_secure,simmons1984prisoners,anderson2002limits}.
At a high level, the interface is:
\begin{equation*}
\label{eq:classic-stego}
\begin{aligned}
s &\gets \mathrm{Embed}(c, m; k),\\
\hat{m} &\gets \mathrm{Extract}(s; k).
\end{aligned}
\end{equation*}
Intuitively, $\mathrm{Embed}$ makes small, hard-to-notice changes to the cover $c$ so that the payload $m$ is hidden in $s$,
and $\mathrm{Extract}$ reverses the process (given $k$) to recover $\hat{m}$.

In classical linguistic steganography, the cover object is simply text. 
Traditional techniques often modify the existing text via substitutions or syntactic changes \cite{provos2003hide}, while neural methods instead \emph{generate} text under constraints that encode a secret \cite{ziegler-etal-2019-neural}.
Arithmetic-coding style constructions can improve capacity and imperceptibility by using the language model distribution directly during generation \cite{shen-etal-2020-near}. 
All these methods, however, share the same methodology: they hide a message (either fully, or in parts) inside another text. 

On the other hand, the approach suggested in \cite{norelli2025llms} is completely different. 
It allows one to deterministically \textit{transform} a message into a new meaningful one, with a completely unrelated meaning. This can be viewed as the purest form of steganography. 
There is no need for a cover object anymore, as the initial message, under the secret key, is entirely transformed into a new unrelated text, becoming the covertext, so to speak. 
More concretely, \CARTS differs from typical bit-encoding approaches in that the payload is itself a natural-language token sequence rather than a low-rate bitstream, enabling one output token per payload token, thus offering \textit{full capacity} encoding. 
Given a fixed model, \CARTS records the payload’s per-position rank trace under an empty context and regenerates a stegotext that realizes the same ranks under a secret keyed context, yielding a length-preserving and deterministically invertible encoding.
A detailed description of this scheme is given in Section \ref{sec:carts-protocol}.

\section{\CARTS protocol}
\label{sec:carts-protocol}
\subsection{The standard protocol}
Fix an autoregressive language model $\mathcal{M}$, and 
let $V_{\mathrm{raw}}$ be its finite tokenizer vocabulary.
We fix two token sets, $V\subseteq V_{\mathrm{raw}}$ and $C\subseteq V_{\mathrm{raw}}$, 
where $V$ is the admissible output vocabulary over which next-token
rankings and deterministic generation are performed, and $C$ is the
context vocabulary used for prompts and previously generated tokens.
We assume $V\subseteq C$, so that generated tokens can be appended to a
context.
We also assume that special tokens or any tokens that are masked by the implementation are 
either removed from $V$ or included with a fixed deterministic masking
convention. All rankings below are rankings over this same admissible
output vocabulary $V$.

Let $N:=|V|$. We write $V^n$ for the set of admissible generated token sequences of
length $n$, and $C^*$ for the set of finite token contexts. For
compactness, write
\[[N]:=\{1,\ldots,N\}, \qquad \mathcal{X}_n:=[N]^n .\]
Here, length is always measured in tokens under the fixed tokenizer
associated with the model.

The protocol described further below runs for exactly $n$ generation steps. If an
end-of-sequence token is retained in $V$, it is treated as an ordinary
token for the purposes of ranking and generation, and it does not cause
early termination.

\begin{defn}[General CARTS protocol]\label{CARTS_def}
Fix a key space $\mathcal{K}\subseteq C^*$. A length-preserving
\CARTS protocol on $V^n$ is a keyed pair of algorithms
\[ E_k,D_k:V^n\longrightarrow V^n, \qquad k\in\mathcal{K}, \]
such that
\[ D_k(E_k(x))=x \]
for every payload token sequence $x\in V^n$ and every key $k\in\mathcal{K}$.

We say that $E_k$ encodes a payload text into a stegotext, and that
$D_k$ decodes the stegotext back into the payload. In \CARTS, the key
$k$ is realized as a secret context or prompt.
\end{defn}

We now describe the Calgacus \cite{norelli2025llms} instantiation of \CARTS. 
For our fixed model $\mathcal{M}$, and a context $c\in C^*$, we write
$p_{\mathcal{M}}(\cdot\mid c)$
for the restricted and renormalized next-token distribution on $V$.
This induces a fixed, deterministic ordering on $V$, by decreasing probability under
$p_{\mathcal{M}}(\cdot\mid c)$, which we denote by 
\[\pi_c:[N]\longrightarrow V. \] 
We assume this ordering to have no ties, as they can be broken deterministically. 
Thus, $\pi_c(j)$ is the token of rank $j$, given the context $c$.

For a token $a\in V$ and a context $c\in C^*$, we define
\[\operatorname{rank}_{\mathcal{M}}(a\mid c)=j \quad\Longleftrightarrow\quad \pi_c(j)=a.\]
As $\mathcal{M}$ is fixed, we often will omit the subscript.

For a sequence $x=(x_1,\ldots,x_n)\in V^n$ and context $c\in C^*$,
define the autoregressive rank trace
\[R_c(x)=\left( \operatorname{rank}(x_i\mid c,x_1,\ldots,x_{i-1}) \right)_{i=1}^n \in \mathcal{X}_n .\]
This is the vector of ranks obtained by reading $x$ from left to right
while updating the context with the previously read tokens.
Conversely, for a rank vector $r=(r_1,\ldots,r_n)\in\mathcal{X}_n$
and context $c\in C^*$, define the autoregressive rank generator
\[G_c(r)=y=(y_1,\ldots,y_n)\in V^n\]
recursively by
\[y_i=\pi_{(c,y_1,\ldots,y_{i-1})}(r_i),\qquad i=1,\ldots,n .\]
Thus, $G_c$ generates one token at a time, always selecting the token
whose current rank is prescribed by $r_i$.

\begin{lemma}[Rank-trace inversion]\label{lem:rank-trace-inversion}
For every context $c\in C^*$, every token sequence $x\in V^n$, and
every rank vector $r\in\mathcal{X}_n$,
\[ R_c(G_c(r))=r, \qquad G_c(R_c(x))=x. \]
\end{lemma}

\begin{proof}
Both identities follow position by position from the definition of
$\pi_c$. At step $i$, the rank trace and the rank generator use the
same updated context $c$ followed by the tokens already processed or
generated. Since $\pi_c$ is a deterministic ordering of the vocabulary,
selecting a token by rank and then measuring its rank gives back the
same rank, and measuring a token's rank and then selecting that rank
gives back the initial token.
\end{proof}

The Calgacus \CARTS encoder and decoder are
\[ E_k(x):=G_k(R_{\varnothing}(x)), \qquad D_k(y):=G_{\varnothing}(R_k(y)), \]
where $\varnothing$ denotes the empty context.
We thus directly have the following correctness result, 
whose proof follows directly from Lemma \ref{lem:rank-trace-inversion}.
\begin{theorem}[Correctness]\label{thm:carts-correctness}
For every $x\in V^n$ and every key $k\in\mathcal{K}$,
\[ D_k(E_k(x))=x. \]
\end{theorem}

The above theorem is a mathematical statement
about a fixed deterministic next-token ordering, which works 
computationally only when the encoder and decoder compute exactly the
same token rankings at every step. 
In practice, this requires identical tokenization, admissible vocabulary, prompt serialization, model weights, masking convention, numerical precision, and tie-breaking rule between sender and receiver; if any of these differ, correctness may fail.

\begin{rmk}[Token perturbations and error propagation]
\label{rmk:error-propagation}
These results assume that the received stegotext is
identical to the transmitted stegotext. The protocol is generally not
robust to token edits. If two stegotexts $y,\tilde{y}\in V^n$ agree up
to position $j-1$ but differ at position $j$, then their rank traces
under the same key $k$ agree up to position $j-1$, but the rank at
position $j$ differs:
\[ R_k(y)_i=R_k(\tilde{y})_i \quad\text{for }i<j, \qquad R_k(y)_j\neq R_k(\tilde{y})_j. \]
For positions $i>j$, there is no general guarantee, because the
autoregressive contexts have diverged. Thus a single token edit can
affect the decoded suffix. Robust variants would require additional
redundancy, synchronization, or error-correction mechanisms. 
We do not address these concerns here, 
as we focus instead on the security concerns and assume that the messages are well transmitted. 
\end{rmk}

\subsection{Generalized two-context rank transforms}
\label{subsec:generalized-carts}

The base Calgacus construction uses the empty context to read the
payload rank trace and the key context to generate the stegotext. The
same formalism supports a more general two-context construction.

Let $a,b\in C^*$ be two contexts, and let
\[ \phi:\mathcal{X}_n\to\mathcal{X}_n \]
be a bijection on the rank space. Define
\[E_{a,b,\phi}(x) := G_b(\phi(R_a(x))), \]
and
\[D_{a,b,\phi}(y):=G_a(\phi^{-1}(R_b(y))).\]

\begin{theorem}[Generalized two-context correctness]
\label{thm:generalized-two-context-carts}
For every $a,b\in C^*$, every bijection
$\phi:\mathcal{X}_n\to\mathcal{X}_n$, and every $x,y\in V^n$,
\[D_{a,b,\phi}(E_{a,b,\phi}(x))=x, \qquad E_{a,b,\phi}(D_{a,b,\phi}(y))=y. \]
\end{theorem}

\begin{proof}
For $x\in V^n$, and $a,b$ and $\phi$ as above, we have 
\[ \begin{aligned}
D_{a,b,\phi}(E_{a,b,\phi}(x))
&= G_a\left( \phi^{-1} \left( R_b\left(G_b(\phi(R_a(x)))\right) \right) \right)\\
&= G_a\left(\phi^{-1}(\phi(R_a(x)))\right)\\
&= G_a(R_a(x))\\
&=x.
\end{aligned} \]
Similarly, for $y\in V^n$, 
$E_{a,b,\phi}(D_{a,b,\phi}(y))=y$.
\end{proof}
\begin{rmk}
The generalized construction above preserves exact invertibility for every
bijection \(\phi:\mathcal{X}_n\to\mathcal{X}_n\). However, the quality
of the generated stegotext depends strongly on the choice of
\(\phi\). The autoregressive model typically assigns high probability
to only a small subset of low-rank tokens at each step. Consequently,
transformations \(\phi\) that map low ranks to substantially larger
ranks will tend to force the generator to select improbable tokens,
often degrading fluency and semantic coherence.

In particular, an arbitrary permutation of \(\mathcal{X}_n\) will
generally not preserve the distributional structure of natural language
rank traces, and the resulting stegotext may appear nonsensical and
unnatural. Practical constructions therefore require transformations
\(\phi\) that preserve, approximately preserve, or otherwise control
rank magnitude and local rank statistics.
\end{rmk} 

The original Calgacus instantiation is recovered by taking
\[a=\varnothing, \qquad b=k, \qquad \phi=\operatorname{id}_{\mathcal{X}_n}. \]
Other choices of $a$, $b$, and $\phi$ can model 
other deterministic transformations of the rank trace, such as prompt-before-payload
variants or rank permutations. 

This more general construction could be used to build more interesting primitives, beyond 
the scope of steganography. 
It also remains an open question to build natural 
generalizations allowing for more than two contexts. 
We leave this task for future work, 
and focus here on the formalization of the security properties the existing CARTS scheme.

\begin{rmk}
It is not necessary, in principle, to use large language models to define a CARTS scheme.
However, autoregressive LLMs provide a natural carrier distribution and a rich conditional ranking structure.
Future developments in generative modeling \cite{goodfellow2014generative} (or alternative probabilistic sequence models) may yield additional instantiations of \CARTS-like protocols with different security and robustness tradeoffs.
\end{rmk}

\section{Security assumptions and protocol properties}
\label{sec:security}

In this section, we investigate several security properties and
computational problems naturally associated with \CARTS schemes.
We formalize and study notions including message equivocation, 
key collisions, and key commutativity, 
and explore their relationships and implications for security

\subsection{Scope and adversarial model}
\label{threat-model}
The classical notion of security in steganography is distributional indistinguishability: the stegotext should be statistically indistinguishable from an ordinary cover object, so that an adversary cannot reliably detect that hidden communication is taking place \cite{zhang2021provably,ding2023discop,de2022perfectly}. In the CARTS setting, this reduces to the question of whether LLM-generated text can be made indistinguishable from ordinary human-written or model-generated text. This question has received substantial attention in recent years \cite{mitchell2023detectgpt,sadasivan2023can}, and will continue to do so as language models become more capable. 
It also stands to reason that, as more research and resources are poured into LLMs, they will get better at mimicking human text, and detection will become significantly harder. 

Rather than revisiting this well-studied question, the present work targets a different and largely unexplored set of security properties specific to the rank-transcoding structure of CARTS. 
These properties --- equivocation, key collisions, context search, and non-commutativity of key-induced maps --- arise naturally from the algebraic and computational structure of the construction, and have not previously been studied in the setting of LLM-based steganography. They lead to questions that are novel from both a machine learning and a cryptographic perspective. 

When it comes to the adversarial model, 
we consider a sender and receiver who share a private key $k$, realized
as a secret context or prompt supplied to a fixed autoregressive
language model $\mathcal{M}$. The sender transforms a payload token
sequence $x\in V^n$ into a generated stegotext $y\in V^n$ of the same
token length, and the receiver deterministically recovers $x$ from $y$
using $k$ and the same model specification. 

The adversary may observe $y$ and is assumed to know the protocol
family, the tokenizer, and the model $\mathcal{M}$. We distinguish
several levels of adversarial access. 
In a \emph{ciphertext-only} setting, the adversary observes only the stegotext $y$. 
In a \emph{known-plaintext} setting, the adversary observes one or more payload-stegotext pairs
$(x_i, y_i)$ with $y_i = E_k(x_i)$ for a fixed unknown key $k$. 
Lastly, in a \emph{chosen-plaintext} setting, the adversary can additionally select payloads and obtain
their encodings under the same fixed unknown key. 

The formal analysis in this paper focuses primarily on 
the known-plaintext model, where the attacker can observe payload-stegotext pairs and knows the protocol family, the tokenizer, and the model $\mathcal{M}$.

\subsection{The rank-coordinate notation} \label{sec:rank_coordinate_notation}
We first introduce the rank-coordinate form of the Calgacus protocol.
This removes any ambiguity between token sequences and rank vectors.

Fix $n$ and let $\mathcal{X}_n=[N]^n$. For every key
$k\in\mathcal{K}$, define
\[F_k:\mathcal{X}_n\longrightarrow\mathcal{X}_n, \qquad F_k(r):=R_{\varnothing}(G_k(r)).\]
Thus $F_k$ maps the empty-context rank representation of a payload to
the empty-context rank representation of the resulting stegotext.

One can easily check that the inverse of $F_k$ is
\[ F_k^{-1}(w)=R_k(G_{\varnothing}(w)). \]
Therefore, for each fixed key $k$, the map $F_k$ is a bijection of $\mathcal{X}_n$.
\begin{prop}[Rank-coordinate conjugacy]\label{prop:rank-coordinate-conjugacy}
For every admissible key $k\in\mathcal{K}$,
\[ F_k = R_{\varnothing}\circ E_k\circ G_{\varnothing}, \]
and equivalently,
\[ E_k = G_{\varnothing}\circ F_k\circ R_{\varnothing}. \]
Moreover,
\[ D_k = G_{\varnothing}\circ F_k^{-1}\circ R_{\varnothing}. \]
Consequently, for any keys $k_1,k_2\in\mathcal{K}$,
\[ E_{k_2}\circ E_{k_1} = G_{\varnothing}\circ (F_{k_2}\circ F_{k_1}) \circ R_{\varnothing}. \]
\end{prop}

\begin{proof}
First, for any $r\in\mathcal{X}_n$,
\begin{align*}
(R_{\varnothing}\circ E_k\circ G_{\varnothing})(r)
&=R_{\varnothing}\left(G_k(R_{\varnothing}(G_{\varnothing}(r)))\right)\\
&=R_{\varnothing}(G_k(r))=F_k(r).
\end{align*}
This proves $F_k=R_{\varnothing}\circ E_k\circ G_{\varnothing}$.
Next, composing on the left by $G_{\varnothing}$ and on the right by
$R_{\varnothing}$ gives
\[E_k=G_{\varnothing}\circ F_k\circ R_{\varnothing}.\]

Using the inverse formula $F_k^{-1}(w)=R_k(G_{\varnothing}(w))$, 
we obtain, for $y\in V^n$,
\begin{align*}
(G_{\varnothing}\circ F_k^{-1}\circ R_{\varnothing})(y)
&=G_{\varnothing}\left(R_k(G_{\varnothing}(R_{\varnothing}(y)))\right)\\
&=G_{\varnothing}(R_k(y)) =D_k(y).
\end{align*}

The composition identity for $E_{k_2}\circ E_{k_1}$ follows by applying
the expression $E_k=G_{\varnothing}\circ F_k\circ R_{\varnothing}$
twice and using $R_{\varnothing}\circ G_{\varnothing}=\operatorname{id}_{\mathcal{X}_n}$. 
\end{proof}
\begin{rmk}[Keys versus key-coordinate vectors]
The formal key is always a token context $k\in C^*$. However, in examples, it is
sometimes convenient to describe a key text by its empty-context rank
coordinates. If $q\in\mathcal{X}_t$, then
\[k=G_{\varnothing}(q)\in V^t\subseteq C^*\]
is the corresponding key context. Thus, $q$ is not itself the key; it is
a coordinate representation of a key context whose tokens lie in the
admissible output vocabulary $V$. 
Key collisions are always
collisions between contexts $k_1,k_2\in\mathcal{K}$, 
but can be displayed through the coordinate vectors $q_1,q_2$.
\end{rmk}

\subsection{Definitions and Computational Problems}

With the above notation and formalism in place, we can now introduce the
main computational problems and security notions associated with a
\CARTS protocol. We begin by defining several forms of message
equivocation, which capture the extent to which an observed stegotext
admits multiple plausible payload interpretations under different keys.
These notions formalize deniability properties of the protocol and will
serve as the foundation for the computational problems introduced later
in this section.

\begin{defn}[Full message equivocation]\label{Def_Full_Messsage_Equivocation}
We say that a \CARTS protocol $(E_k, D_k)$ satisfies \textit{full message equivocation} on
$V^n$ if, for every observed stegotext $y\in V^n$ and every payload
$x\in V^n$, there exists a key $k\in\mathcal{K}$ such that
\[ D_k(y)=x. \]
\end{defn}

\begin{defn}[Weak $\ell$-message equivocation]
\label{Def_Weak_Messsage_Equivocation}
For an integer $\ell\leq |V|^n$, 
we say that a \CARTS protocol $(E_k, D_k)$ satisfies 
\textit{weak $\ell$-message equivocation} on $V^n$ if, for every
observed stegotext $y\in V^n$,
there exist distinct payloads
\[x_1,...,x_\ell\in V^n\]
and keys
\[k_1,...,k_\ell\in\mathcal{K}\]
such that
\[D_{k_i}(y)=x_i,\qquad \forall i=1,...,\ell .\]
\end{defn}

Message equivocation (whether it is full or weak) can also be referred to as encryption deniability (in the spirit of \cite{canetti1997deniable,durmuth2011deniable}). 
Intuitively, this property means that a given key $k$ does not uniquely determine a single underlying message. Instead, multiple plausible payloads may correspond to the same observed covertext when different keys are used. Consequently, even if a computationally unbounded adversary intercepts $k$, they cannot definitively determine which message was encoded, since several alternative messages remain consistent with the observed data under different keys. A party may thus plausibly claim that a different message was intended. 
Moreover, full message equivocation implies weak $\ell$-message equivocation (for $\ell$ within the bound of existing messages in $|V|^n$).

\begin{rmk}
The practical relevance of message equivocation depends strongly on the
threat model and the surrounding social or political environment.
Indeed, in highly coercive settings, such as under an authoritarian
regime, the mere existence of a key capable of explaining a stegotext as
corresponding to an incriminating or ``undesirable'' payload may itself
be sufficient to justify retaliation, regardless of whether alternative
benign explanations also exist. In such contexts, message equivocation
alone may therefore provide limited practical protection.

By contrast, in more typical institutional settings governed 
by stronger procedural protections,
where individuals are not presumed guilty solely on the basis of a 
possible incriminating interpretation, equivocation can provide a
meaningful form of deniability. In such cases, the existence of a
plausible alternative payload consistent with the same observed
stegotext may be sufficient to prevent definitive attribution.
Consequently, the operational significance of equivocation depends not
only on the mathematical properties of the protocol, but also on the
standards of evidence and coercion present in the intended application
environment.
\end{rmk}

For completeness, we also introduce a restricted variant of message
equivocation in which the admissible payloads, observed stegotexts, and
keys are constrained to prescribed subsets. This formulation is useful
in practical settings where only certain payloads are considered
plausible, only certain stegotexts are observable, or only a restricted
family of keys is operationally realistic.

\begin{defn}[Restricted-support equivocation]
Let $\mathcal{P}_{\mathrm{adm}}\subseteq V^n$ be a set of admissible payloads, let
$\mathcal{Y}_{\mathrm{adm}}\subseteq V^n$ be a set of admissible observed
stegotexts, and let $\mathcal{K}_{\mathrm{adm}}\subseteq\mathcal{K}$ be
a set of admissible or plausible keys. The protocol has \textit{restricted-support full message equivocation} relative to
$(\mathcal{P}_{\mathrm{adm}},\mathcal{Y}_{\mathrm{adm}},\mathcal{K}_{\mathrm{adm}})$ if, for every
$x\in\mathcal{P}_{\mathrm{adm}}$ and every $y\in\mathcal{Y}_{\mathrm{adm}}$, there exists
$k\in\mathcal{K}_{\mathrm{adm}}$ such that
\[D_k(y)=x.\]

For an integer $\ell\leq |\mathcal{P}_{\mathrm{adm}}|$, it has \textit{restricted-support weak $\ell$-message equivocation} relative to
$(\mathcal{P}_{\mathrm{adm}},\mathcal{Y}_{\mathrm{adm}},\mathcal{K}_{\mathrm{adm}})$ if, for every
$y\in\mathcal{Y}_{\mathrm{adm}}$, there exist distinct
$x_1,...,x_\ell\in\mathcal{P}_{\mathrm{adm}}$ and keys
$k_1,...,k_\ell\in\mathcal{K}_{\mathrm{adm}}$ such that
\[D_{k_i}(y)=x_i, \qquad \forall i=1,...,\ell .\]
\end{defn}

Having introduced the relevant security definitions, we now
turn to the computational problems naturally associated with a \CARTS
protocol. These problems formalize the algorithmic tasks underlying the
security and deniability properties discussed above.
We work with the same setup as Section \ref{sec:carts-protocol}, and we assume that our large language model $\mathcal{M}$ is fixed.

For unconditional properties, we allow the key space
$\mathcal{K}\subseteq C^*$ to be arbitrary. For algorithmic search
problems, however, the key domain must be specified more concretely, 
for instance, by restricting to a finite admissible key set
$\mathcal{K}_{\mathrm{adm}}\subseteq\mathcal{K}$, or by giving a specified key-generation distribution with an explicit search budget. 
Without such a restriction, non-existence over all of $C^*$ is not a
finite computational task.

For convenience, we use the rank-coordinate notation introduced in Section \ref{sec:rank_coordinate_notation} to emphasize the \textit{action} of the key $k$ on the messages, viewed as rank vectors.

\begin{defn}[Key collision]\label{Def_Key_Collision}
Fix a model $\mathcal{M}$, a key space $\mathcal{K}$, and a length
$n$. For a rank vector $r\in\mathcal{X}_n$, two distinct keys
$k_1,k_2\in\mathcal{K}$ form a \textit{key collision} if
\[F_{k_1}(r)=F_{k_2}(r).\]
Equivalently, the same payload rank vector $r$ 
maps to the same stegotext rank vector under two different keys. 
\end{defn}

\begin{prob}[The key collision problem]\label{Prob_Key_Collision}
Given $r\in\mathcal{X}_n$ and an admissible key set
$\mathcal{K}_{\mathrm{adm}}\subseteq\mathcal{K}$, find distinct keys
$k_1,k_2\in\mathcal{K}_{\mathrm{adm}}$ such that
\[F_{k_1}(r)=F_{k_2}(r),\]
or determine that no such pair exists within
$\mathcal{K}_{\mathrm{adm}}$.
\end{prob}

\begin{prob}[The promise context search problem]
\label{Prob_Standard_Vectorization}
Given $r,w\in\mathcal{X}_n$ and an admissible key set
$\mathcal{K}_{\mathrm{adm}}\subseteq\mathcal{K}$, with the promise
that there exists $k\in\mathcal{K}_{\mathrm{adm}}$ satisfying
\[ F_k(r)=w, \]
find one such key $k$.
\end{prob}

Problem \ref{Prob_Standard_Vectorization} is the most direct known-message key-search problem. It is very close
in spirit to the \textit{vectorization problem} (see Problem 6 in \cite{couveignes2006hard})
because one is searching for a key
that transports one rank vector to another. 
The main difference with the usual vectorization problem is that the key space does not carry a natural algebraic composition law (as in \cite{alamati2020cryptographic}), and the family $\{F_k\}$ is not assumed to form a group action. The context-search problem is therefore a transport problem on rank-vector space without the algebraic structure that underlies the vectorization setting.

\begin{prob}[The unconditional context search problem]\label{Prob_Uncond_Vectorization}
Given $r,w\in\mathcal{X}_n$ and an admissible key set
$\mathcal{K}_{\mathrm{adm}}\subseteq\mathcal{K}$, find a key 
$k\in\mathcal{K}_{\mathrm{adm}}$ such that
\[ F_k(r)=w. \]
\end{prob}

\begin{prob}[The unconditional partial $\ell$-context search problem]
\label{Prob_Uncond_Partial_Vectorization}
Given $w\in\mathcal{X}_n$ and an admissible key set
$\mathcal{K}_{\mathrm{adm}}\subseteq\mathcal{K}$, find $\ell$ rank vectors $r_1, ..., r_\ell \in\mathcal{X}_n$ and keys
$k_1, ..., k_\ell \in\mathcal{K}_{\mathrm{adm}}$ such that
\[ F_{k_i}(r_i)=w. \]
\end{prob}

Note, of course, that for an arbitrary CARTS protocol, a solution to Problems \ref{Prob_Uncond_Vectorization} or \ref{Prob_Uncond_Partial_Vectorization} need not exist. 
Furthermore, it is clear that Problem \ref{Prob_Uncond_Partial_Vectorization} is not harder than Problem \ref{Prob_Uncond_Vectorization}.
With the above notions in hand, we can discuss message equivocation for a \CARTS protocol. 
This will be done in the following subsection.

Finally, for completeness, we mention a more general version of Problem \ref{Prob_Uncond_Vectorization} that will be addressed further in subsequent work, when building a different primitive.

\begin{prob}[Pre-image resistance problem]\label{Prob_Binding}
Given $w\in\mathcal{X}_n$ and an admissible key set
$\mathcal{K}_{\mathrm{adm}}\subseteq\mathcal{K}$, find a key 
$k\in\mathcal{K}_{\mathrm{adm}}$ and a rank vector $r\in\mathcal{X}_n$ such that
\[ F_k(r)=w. \]
\end{prob}

Although this problem looks similar to Problem \ref{Prob_Uncond_Vectorization}, the {key} difference is that the message $r$ is not fixed here. 
More generally, all the above definitions and problems look similar but differ in very subtle ways, as they target different privacy concerns and applications.

\begin{rmk}
\begin{enumerate}[label=(\arabic*), leftmargin=*]
\item 
An unconstrained preimage-search problem of the form ``given $w$, find
$r$ and $k$ such that $F_k(r)=w$'' is trivial whenever the empty context
is admissible: take $r=w$ and $k=\varnothing$. 
One can thus consider additional constraints, such as $r\neq w$, non-empty
keys, plausible keys, plausible payloads, or bounded key length.
\item 
If the empty context $\varnothing$ is included in the admissible key
space $\mathcal{K}$, then $F_{\varnothing}(r)=r$ 
for every $r\in\mathcal{X}_n$. 
Therefore, every stegotext $y$ has the trivial
explanation $E_{\varnothing}(y)=y$.
Consequently, the unconditional partial $1$-context search problem 
(and hence, weak $1$-message equivocation, as we will see in Section \ref{sec:Security_Properties_of_the_Protocol}) is trivial whenever the
empty key is admissible. Nontrivial equivocation claims 
should therefore impose additional constraints. 
\end{enumerate}
\end{rmk}

Lastly, we note that all of the computational problems 
introduced above can be defined for a
general \CARTS protocol in terms of the encoding and decoding maps
$E_k$ and $D_k$ alone. The rank-coordinate formulation via
$F_k:\mathcal{X}_n \to \mathcal{X}_n$ is specific to the Calgacus
instantiation and its induced rank representation, and is introduced
primarily for convenience of analysis and notation. The underlying
problem structure, however, is independent of this representation.

\subsection{Security Properties of the Protocol} \label{sec:Security_Properties_of_the_Protocol}

\subsubsection{Message Equivocation}

\begin{theorem}
\label{thm:equivocation-vectorization}
A \CARTS protocol has full message equivocation on
$V^n$ if and only if 
the unconditional context search problem 
(Problem~\ref{Prob_Uncond_Vectorization}) is solvable for all $r,w\in\mathcal{X}_n$. 

Similarly, for an integer $\ell\leq |\mathcal{X}_n|$, 
it has weak $\ell$-message equivocation on $V^n$
if and only if the unconditional partial $\ell$-context search problem
(Problem~\ref{Prob_Uncond_Partial_Vectorization}) is solvable for every $w\in\mathcal{X}_n$.
\end{theorem}

\begin{proof} 
The result follows directly from unpacking the definitions. 
For notational convenience, we present the proof for the Calgacus construction; the same argument applies to any \CARTS protocol. 
Recall that $R_{\varnothing}:V^n \to \mathcal{X}_n$ is a bijection with inverse $G_{\varnothing}$, and that $F_k = R_{\varnothing} \circ E_k \circ G_{\varnothing}$. 
Let $x,y \in V^n$ and define $r = R_{\varnothing}(x)$ and $w = R_{\varnothing}(y)$. Then, 
\[D_k(y)=x \;\Longleftrightarrow\; 
G_{\varnothing}(R_k(y)) = G_{\varnothing}(r) \;\Longleftrightarrow\; R_k(y)=r. \]
Substituting $y = G_{\varnothing}(w)$ and using the definition of $F_k$, this is equivalent to
$F_k(r)=w$. 
This establishes the equivalence between the full-message equivocation condition and the existence of a key solving the unconditional context search problem in rank space.

The weak $\ell$-message statement follows by applying the same argument to $\ell$ distinct pairs.
\end{proof}

We now state a direct consequence of the above characterization, which extends to the restricted-support setting via the same rank-coordinate transformation. 

\begin{corollary}
Let $\mathcal{R}_{\mathcal{P}}:=R_{\varnothing}(\mathcal{P}_{\mathrm{adm}})$ and $\mathcal{R}_{\mathcal{Y}}:=R_{\varnothing}(\mathcal{Y}_{\mathrm{adm}})$, 
then the Calgacus \CARTS protocol 
has restricted-support full message equivocation relative
to $(\mathcal{P}_{\mathrm{adm}},\mathcal{Y}_{\mathrm{adm}},\mathcal{K}_{\mathrm{adm}})$ if and
only if
\[ \forall r\in\mathcal{R}_{\mathcal{P}}, \forall w\in\mathcal{R}_{\mathcal{Y}}, 
\quad \exists k\in\mathcal{K}_{\mathrm{adm}}  \quad\text{such that}\quad F_k(r)=w. \]
Similarly, it has restricted-support weak $\ell$-message equivocation
if and only if, for every $w\in\mathcal{R}_{\mathcal{Y}}$, there exist
distinct $r_1,\ldots,r_\ell\in\mathcal{R}_{\mathcal{P}}$ and keys
$k_1,\ldots,k_\ell\in\mathcal{K}_{\mathrm{adm}}$ such that $F_{k_i}(r_i)=w$ for all $i=1,...,\ell$.
\end{corollary}

\subsubsection{Key collision} \label{subsection_Key_Collisions}
Let us now discuss Problem~\ref{Prob_Key_Collision} more closely. 
Key collisions do not affect correctness: for any fixed
key $k$, the decoder still inverts the encoder by
Theorem~\ref{thm:carts-correctness}. They are however interesting from a
security standpoint. 

Suppose an adversary knows a payload-stegotext pair $(x,y)$, or
equivalently the rank-coordinate pair
\[ r=R_{\varnothing}(x), \qquad w=R_{\varnothing}(y). \]
Denote the set of keys consistent with this observation by 
\[ \mathsf{ValidKeys}(r,w):=\{k\in\mathcal{K}:F_k(r)=w\}. \]
If $|\mathsf{ValidKeys}(r,w)|>1$, then the key is not
information-theoretically identifiable (even for a computationally unbounded adversary) 
from this single transcript, and thus cannot be used to decode future messages from the sender. 
In practice however, this is slightly weaker than saying that future communications are secure, 
as additional known plaintexts, side information about the key, or a prior
over plausible prompts may reduce (or even eliminate) the candidate set.
Indeed, as additional payload-stegotext pairs are observed, 
the candidate key set can only shrink, since each new transcript 
imposes an additional constraint of the form $F_k(r_i) = w_i$.

For a finite key set $\mathcal{K}'\subseteq\mathcal{K}$, define the
collision multiplicity
\[ \mu_{\mathcal{K}'}(r,w) := \left|\{k\in\mathcal{K}':F_k(r)=w\}\right |. \]
Thus, $\mu_{\mathcal{K}'}(r,w)=0$ means that $w$ is unreachable from
$r$ using keys in $\mathcal{K}'$, $\mu_{\mathcal{K}'}(r,w)=1$ means
that the key is identifiable within $\mathcal{K}'$ from this single
pair, and $\mu_{\mathcal{K}'}(r,w)>1$ means that the pair admits a key
collision within $\mathcal{K}_{\mathrm{adm}}$.

If keys are generated by a distribution $\mathsf{KeyGen}$, 
a distributional collision probability for fixed $r$ is
\[ \operatorname{Coll}(r) 
= \Pr_{k_1,k_2\leftarrow \mathsf{KeyGen}} \left[ k_1\neq k_2 \ \wedge\ F_{k_1} r)=F_{k_2}(r) \right]. \]
This quantity is different from the existence of a collision. A single
collision may exist while having negligible probability under the key
distribution, whereas a large collision probability indicates that key
\textit{non-identifiability} is common for the chosen key-generation procedure.

Here, we present an initial investigation of the key-collision behavior in the Calgacus \CARTS protocol, leaving more extensive analysis for future work. 
Preliminary experiments suggest that key collisions are unlikely but can still occur in
specific finite key searches. 

All numeric ranks in the following example are computed using the Llama-3-8B-Instruct GGUF model. Since the underlying sentences (token sequences) are explicitly provided, the example can be reproduced (with potentially different values for $q_1$ and $q_2$). 
\begin{ex} \label{ExampleOfCollision}
Using a specific version of the Llama-3-8B-Instruct GGUF model, 
we can observe that the following two key-coordinate vectors 
\begin{align*}
q_1&=[88599,1087,1,1,1], \\
q_2&=[88599,1087,47,646,1]
\end{align*}
both map the rank-coordinate vector $r:=[1, 1, 1, 1, 1]$ to $w:=[126444, 1, 3736, 2, 4]$, 
thus producing a collision. 
Letting 
\[k_1:=G_{\varnothing}(q_1) \qquad\text{and}\qquad k_2:=G_{\varnothing}(q_2) \]
denote the corresponding key contexts, we write 
\[ F_{k_1}(r)= F_{k_2}(r)=w. \]
While this is \textit{a priori} intriguing, examining the corresponding sentences reveals only a minor (typographical) difference between the keys' token sequences:
\begin{align*}
\begin{split}
G_{\varnothing}(q_1)&=\text{\ttfamily The quick brown fox jumps}\\
G_{\varnothing}(q_2)&=\text{\ttfamily The quick brow fox jumps}\\
G_{\varnothing}(w)&=\text{\ttfamily over the lazy dog.}
\end{split}
\end{align*}
Notice the typo in $G_{\varnothing}(q_2)$, and that all vectors above have the same length, corresponding to sentences with an equal number of tokens (the period at the end of $G_{\varnothing}(w)$ is one of the tokens).
\end{ex}

In view of our previous example, one can ask if two distinct but nearly identical keys will 
encode a message into two nearly identical stegotexts (and decode a stegotext into two nearly identical messages). 
In general, for a key collision to provide meaningful security for the sender, 
the alternative keys must decode the new stegotexts to \textit{semantically distinct messages}. 
If all collisions correspond only to minor variations (e.g., typos or trivial rewordings), an adversary could easily recover the original key by correcting these minor differences. But this only makes sense if such a correction only affects stegotexts minimally. 

Additional considerations are whether the alternative keys are plausible 
under the adversary's key prior, how large the candidate set $\mathsf{ValidKeys}(r,w)$ is, and whether the
candidate keys remain consistent across additional transcripts.

This motivates natural empirical questions: how often do nontrivial
collisions occur, how large can $|\mathsf{ValidKeys}(r,w)|$ become for a
given pair $(r,w)$, and how quickly do candidate keys separate when
tested on additional messages? We investigate these questions in
Section~\ref{EmpiricalExplorations_Keycollision1}.

Although the two keys in Example~\ref{ExampleOfCollision} are
semantically very close, this does not by itself determine whether the
collision is cryptographically meaningful. A small textual perturbation
in a key (as in Example~\ref{ExampleOfCollision}) may either remain localized 
(leading to a collision for only a single payload) or may lead to substantially
different behavior on a multitude of payloads. We therefore distinguish
\emph{local key collisions}, which occur only for a specific 
$r$, from stronger forms of key ambiguity in which the same key pairs 
remain indistinguishable globally, across many payloads. We study this stability
question empirically in Section~\ref{EmpiricalExplorations_Keycollision2} 
and see that even though the keys $k_1$ and $k_2$ are nearly identical, and agree 
on the input $[1,1,1,1,1]$, they disagree on other random inputs. 

\begin{rmk}
If one does not restrict the lengths of the keys, nor require them to be equal, 
it becomes considerably easier to find collisions, 
for instance, by looking for extremely common phrases that end with the same words, such as: ``\textit{the capital city of Italy is \textbf{Rome}}'' and ``\textit{all roads lead to \textbf{Rome}}''. 
For this reason, collision
experiments should specify the admissible key set
$\mathcal{K}_{\mathrm{adm}}$, the key-length policy, and the key-generation
distribution. The relevant empirical quantities are not merely whether
a collision exists, but the size of the candidate fiber
$\mu_{\mathcal{K}_{\mathrm{adm}}}(r,w)$ and the probability that independently
generated keys collide.
\end{rmk}

\subsubsection{Finite cardinality effects}
We now revisit the previous problem from a new perspective. 
Fix a payload rank
vector $r\in\mathcal{X}_n$ and let 
$\mathcal{K}_{\mathrm{adm}}:= \mathcal{X}_n = [N]^n$, 
that is, we assume that the keys have the same length $n$ as the input message. 
Define
\begin{align*}
f_r:\mathcal{X}_n&\longrightarrow\mathcal{X}_n,\\
f_r(k)&:=F_k(r).
\end{align*}
The fiber
\[ f_r^{-1}\{w\}=\{k\in\mathcal{X}_n:F_k(r)=w\} \]
is the set of keys that map payload rank vector $r$ to stegotext rank
vector $w$.

This allows us to rephrase the existence of key collisions for a given $r\in\mathcal{X}_n$ 
(Problem \ref{Prob_Key_Collision}) in terms of $f_r$ being \textit{non-injective}. 
On the other hand, we see that the unconditional context search problem 
(Problem \ref{Prob_Uncond_Vectorization}) has a solution 
if and only if the functions $f_r$ are \textit{surjective} for all $r$. 
In particular, we have the following result. 

\begin{prop}\label{prop:finite-key-cardinality}
For fixed $r\in\mathcal{X}_n$, the following are equivalent:
\begin{enumerate}
    \item $f_r$ is injective.
    \item $f_r$ is surjective.
    \item There are no key collisions for this fixed $r$ within
    $\mathcal{K}_{\mathrm{adm}}$.
    \item For every $w\in\mathcal{X}_n$, there exists a key
    $k\in\mathcal{K}_{\mathrm{adm}}$ such that $F_k(r)=w$.
\end{enumerate}
Consequently, if a key collision exists for some output $w$, then
$f_r$ is not injective and hence not surjective. Therefore, at least
one output $w'\in\mathcal{X}_n$ is unreachable from $r$ using keys in
$\mathcal{K}_{\mathrm{adm}}$. However, for the colliding output $w$ itself, the
context-search instance $F_k(r)=w$ is solvable and has more than one
solution.
\end{prop}

\begin{proof}
The equivalence between injectivity and surjectivity follows from the
fact that $f_r$ is a map between two finite sets of equal cardinality.
Condition 3 is the statement that no two distinct keys map $r$ to the same output,
which is injectivity. 
Condition 4 is exactly the statement that every output has a 
preimage under $f_r$, which is equivalent to surjectivity. 
The final claim follows immediately: a collision
makes $f_r$ non-injective, and therefore non-surjective, but the
particular output at which the collision occurs has at least two
preimages.
\end{proof}

Proposition~\ref{prop:finite-key-cardinality} and
Corollary~\ref{cor:key-length-cardinality} (below) highlight 
the natural tension between the existence of key collisions and the existence of a solution to the context search problem, i.e. the tension between Problems \ref{Prob_Key_Collision} and \ref{Prob_Uncond_Vectorization}. 
For a fixed payload rank vector $r$ and an
equal-size finite key space, one cannot have both a collision-free map
and missing outputs. 

\begin{corollary}[Key-length cardinality regimes]
\label{cor:key-length-cardinality}
Let $\mathcal{A}\subseteq C$ be an admissible key alphabet of size
$M:=|\mathcal{A}|$, and let $\mathcal{K}_t:=\mathcal{A}^t$ 
be the set of all key contexts of length $t$ over $\mathcal{A}$. Fix
$r\in\mathcal{X}_n$, and define
\[ f_r:\mathcal{K}_t\to\mathcal{X}_n, \qquad f_r(k)=F_k(r). \]
Then:
\begin{enumerate}
    \item If $M^t<N^n$, then $f_r$ cannot be surjective. Hence some
    output rank vectors are unreachable from $r$ using keys in
    $\mathcal{K}_t$.
    \item If $M^t>N^n$, then $f_r$ cannot be injective. Hence at least
    one key collision exists for this fixed $r$ within $\mathcal{K}_t$.
    \item If $M^t=N^n$, then injectivity, surjectivity, unique
    reachability, and absence of key collisions are equivalent, as in
    Proposition~\ref{prop:finite-key-cardinality}.
\end{enumerate}
In the special case $M=N$, the regimes are $t<n$, $t=n$, and $t>n$.
Thus, shorter-than-message key spaces cannot reach every output for a
fixed payload rank vector, while longer-than-message key spaces force
collisions by cardinality alone.
\end{corollary}

Finally, we briefly mention that one can reinterpret the above scheme in terms of \textit{set actions}. 
Let $X:=\mathcal{X}_n$, then the set of keys $\mathcal{K}_{\mathrm{adm}}$ 
acts on $X$ via $k \star r := F_k(r) \in X$. 
Then, the unconditional context search problem (Problem \ref{Prob_Uncond_Vectorization}) 
has a solution if and only if the set action is transitive. 
Although set actions are not often studied—since, unlike group actions, they lack algebraic structure—they provide a natural viewpoint in our setting. 
Indeed, the LLM construction does not naturally induce a group structure, so this is the most appropriate formalism available.  
Moreover, due to the lack of structure and algebraicity behind this action, 
there aren't any natural approaches or interpretations that could lead to potential algebraic attacks or exploits. 

\subsubsection{Non-commutativity} \label{Subsection_Noncommutativity}

For completeness, we record an additional algebraic question related to
possible instance-generation attacks. In many cryptographic problems, a
single hard instance can be transformed into multiple related instances
with the same hidden secret. If enough such related instances are
available, the original recovery problem may become easier.

A concrete example of this phenomenon 
appears in the lattice isomorphism problem (LIP) \cite{ducas2022lattice}. In one
formulation, the task is to find a unimodular matrix $A$ such that
\[ Q'=A^T Q A \]
from input $(Q,Q')$. 
The authors of \cite{benvcina2024properties} explain that if a second unimodular matrix $U$ commutes with
$A$, then one can form a second instance of the LIP problem, with the same solution $A$, by letting
\[ \bar{Q}:=U^T Q U, \qquad \bar{Q}':=U^T Q' U. \]
Indeed, by commutativity, we have $\bar{Q}'=A^T \bar{Q}A$. 
They then show that if we have enough such LIP instances, with the same secret $A$, we can recover $A$. This in turn creates the requirement that it should not be easy for the attacker to find too many matrices that commute with $A$. 

This example motivates asking whether an analogous instance-generation structure 
applies to the context search problems (see Problems \ref{Prob_Standard_Vectorization} and \ref{Prob_Uncond_Vectorization}), which underlie the privacy of the \CARTS system. 
In other words, we are asking whether $$E_{k_1}(E_{k_2}(m)) = E_{k_2}(E_{k_1}(m))$$ holds for certain messages $m$ and keys $k_1, k_2$.
We illustrate this question using the diagram in Figure~\ref{Fig:commDiagram}. The diagram highlights two possible paths from $m$ to a doubly encoded message. The question is whether these two paths can lead to the same result, i.e., whether the diagram commutes.

\begin{figure}[h] 
\centering
\begin{tikzpicture}
\node (m) at (0,3) {$m$};
\node (Ek1m) at (4.5,3) {$E_{k_1}(m)$};

\node (Ek2m) at (0,0) {$E_{k_2}(m)$};

\node (topright) at (4.6,0.6) {$E_{k_2}(E_{k_1}(m))$};
\node (bottomright) at (3.4,-0.6) {$E_{k_1}(E_{k_2}(m))$};
\node[above] (QuestMark) at (3.9,-0.1) {$?$};

\draw[->] (m) -- node[above] {$E_{k_1}$} (Ek1m);
\draw[->] (m) -- node[left] {$E_{k_2}$} (Ek2m);

\draw[->] (Ek1m) -- node[right] {$E_{k_2}$} (topright);
\draw[->] (Ek2m) -- node[below] {$E_{k_1}$} (bottomright);
\draw[double distance=1.2pt, line width=0.4pt] (bottomright) -- (topright);
\end{tikzpicture}
\caption{\textbf{Commutativity test for the encoding operators.} Starting from a message $m$, the diagram compares the two compositions $E_{k_2} \circ E_{k_1}$ and $E_{k_1} \circ E_{k_2}$. The diagram is said to commute if both paths yield the same output for all messages $m$. }
\label{Fig:commDiagram}
\end{figure}
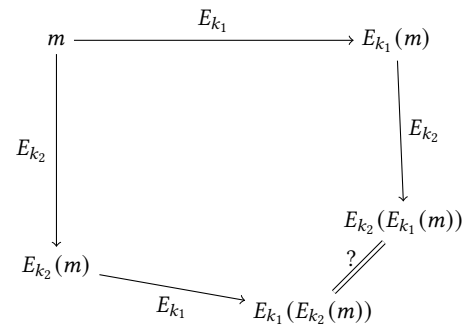 

If such commutativity were to hold (for a sufficiently large family of keys), then the order in which encoding operators are applied would be irrelevant and this symmetry can be used to generate multiple related instances of the context search problem, sharing the same secret key.
Namely, given a single context search problem instance $(m,E_k(m))$, where one is looking for $k$, 
we would be able to generate multiple instances 
\[ \mathcal{I}:=\{ \big(m_i, E_k(m_i)\big): i=1,2,... \}, \]
all with the same secret $k$, 
where $m_i:= E_{k_i}(m)$ for keys $k_1, k_2, ...$ satisfying the commutativity property 
\[E_{k_i}(E_{k}(m)) = E_{k}(E_{k_i}(m)). \] 
This additional information, namely the knowledge of the set $\mathcal{I}$ above, instead of simply the single instance $(m,E_k(m))$, could significantly reduce the hardness of recovering the secret key $k$.

We investigate this commutativity question further in Section \ref{EmpiricalExplorations_Noncommutativity}, from an experimental point of view, to gain a better initial understanding of the situation, and leave the theoretical study of this question for future work. 
Furthermore, it would be interesting to study, from a theoretical point of view, what assumptions one must put on a language model, in order to obtain specific commuting keys, and from a practical point of view  the number of commuting key pairs needed to recover the secret $k$. 


\section{Empirical Explorations}
\label{sec:EmpiricalExplorations} 

The theoretical framework developed in the preceding sections raises many 
concrete empirical questions such as: do key collisions arise in practice, 
are the key-induced maps far from commuting, and how brittle is the 
protocol to channel noise? This section reports five experiments that 
investigate these questions and more, under a fixed Llama 3 8B configuration.

\begin{table*}[t]
\centering
\footnotesize
\caption{Empirical experiments summary.}
\label{tab:master-empirical-design}
\begin{tabular}{p{0.23\textwidth}p{0.6\textwidth}p{0.05\textwidth}}
\hline
\textbf{Experiment} &
\textbf{Description} &
\textbf{Runtime} \\
\hline

Implementation correctness &
40 payload-key encoding-decoding tests; forward and reverse recovery checked. &
607.6 s \\

Key-collisions and finite key search &
16 collision-search transcripts; 60 finite keys searched per transcript,
for \(16\times 60=960\) finite-key evaluations. &
3727.7 s \\

Collision stability across new payloads &
Hand-crafted collision from Example~\ref{ExampleOfCollision}, tested 
across 9 additional rank vectors. &
105.4 s \\

Non-commutativity of encoding maps &
36 sampled key pairs; 8 payload rank vectors tested per key pair. &
2265.8 s \\

Robustness to token perturbations &
20 base stegotexts; 80 length-preserving perturbation cases. &
561.9 s \\

\hline
\end{tabular}
\end{table*}

\subsection{Experimental Configuration and Common Metrics}
\label{EmpiricalExplorations_Setup}
All experiments use the \texttt{llama3\_8b\_q4\_k\_m} (Q4\_K\_M GGUF) model \cite{grattafiori2024llama} 
via \texttt{llama-cpp-python}, with $n_{\mathrm{ctx}}=4096$, 
\texttt{logits\_all=true},
and CPU-only inference (\texttt{n\_gpu\_layers=0}, default CPU-threading). 
The software environment was Python \texttt{3.12.3} on Linux \texttt{6.17.0-22-} \texttt{generic-x86\_64}.
Non-empty key contexts are tokenized with \texttt{add\_} \texttt{bos=True}, after which the initial BOS token is dropped to avoid duplicating the beginning-of-sequence
marker; empty contexts use the model BOS token as the minimal autoregressive context. 
Payload and stegotext display strings were tokenized by prepending a
leading space, tokenizing with \texttt{add\_bos=True}, and then
discarding the initial BOS token. This convention matches the
implementation used for both encoding and decoding. The prefix was empty in this run.
Ranks are 1-indexed and computed from full-vocabulary logits, 
sorted by decreasing logit value with ties broken by increasing token id. 
No top-$k$, top-$p$, or additional experiment-level vocabulary restriction was applied. 
All experiments share the same prompt serialization, masking convention, ranking rule, and tie-breaking rule. The random seed was fixed to $123$.

The run used 24 natural-language payloads of 4 to 9 tokens, 60 prompt keys, 40 payload-key pairs for
correctness experiments, 16 transcripts for the collision search,
36 key pairs for the non-commutativity experiment, and 80 perturbation cases for the
robustness experiment. 
Table~\ref{tab:master-empirical-design}
summarizes the experiments and their runtimes.


\subsection{Experiment 1: Implementation Correctness}
\label{EmpiricalExplorations_Correctness} 

Theorem~\ref{thm:carts-correctness} guarantees exact payload 
recovery under idealized deterministic assumptions, but this guarantee 
is only as strong as the implementation's fidelity to those assumptions. 
In practice, subtle discrepancies in tokenization conventions, 
BOS token handling, numerical precision, or masking behavior 
could cause silent failures even when the mathematical proof is correct. 
This experiment verifies that the implementation realizes the theorem's 
assumptions exactly under the tested configuration.

For \(40\) sampled payload-key pairs, we computed $y_i = E_{k_i}(x_i)$, 
decoded $\hat{x}_i = D_{k_i}(y_i)$, and verified $\hat{x}_i = x_i$. 
We also tested the reverse direction, verifying $E_{k_i}(D_{k_i}(y_i)) = y_i$.

\begin{table}[H]
\centering
\caption{Implementation correctness results.}
\label{tab:exp1-correctness}
\begin{tabular}{l c}
\hline
\textbf{Check} & \textbf{Successes} \\
\hline
\(D_k(E_k(x))=x\) & \(40/40\) \\
\(E_k(D_k(y))=y\) & \(40/40\) \\
\hline
\end{tabular}
\end{table}

\medskip
\noindent\textbf{Results.}
As shown in Table~\ref{tab:exp1-correctness}, 
exact recovery succeeded in all \(40/40\) cases in both directions,
confirming that the implementation correctly realizes the deterministic 
rank-transcoding mechanism under the tested Llama 3 8B configuration.

\subsection{Experiment 2: Key-collisions and finite key search}
\label{EmpiricalExplorations_Keycollision}
\label{EmpiricalExplorations_Keycollision1}
\label{EmpiricalExplorations_ContextSearch}

This experiment investigates key collisions via exhaustive search over 
a finite key set, which is the simplest possible attack on the 
promised context search problem (Problem~\ref{Prob_Standard_Vectorization}). 

For the purpose of this experiment, 
we generated a finite admissible key set $\mathcal{K}_{\mathrm{adm}}=\{k_1,\ldots,k_{60}\}$ 
containing \(60\) prompt keys. The keys were constructed from a small
set of natural-language seed prompts together with controlled local
variants. The seed prompts covered mundane topics such as baking bread,
mountain villages, forest animals, baseball practice, 
neutral product descriptions, and the phrase
\texttt{The quick brown fox jumps}. From these seeds, we generated
near-duplicate and template-style variants, including one-character
typos, one-token deletions, plural or singular substitutions,
punctuation changes, and short suffix additions.
The composition of the resulting key set $\mathcal{K}_{\mathrm{adm}}$ is given in Table~\ref{300416p}. 
\begin{table}[H]
\centering
\caption{Composition of the finite admissible key set.}
\label{300416p}
\begin{tabular}{lc}
\hline
\textbf{Key category} & \textbf{Count} \\
\hline
Seed or near-duplicate prompts & 7 \\
One-character typo variants & 7 \\
One-token deletion variants & 7 \\
Plural or singular variants & 7 \\
Punctuation variants & 14 \\
Short suffix variants & 18 \\
\hline
\textbf{Total} & \textbf{60} \\
\hline
\end{tabular}
\end{table}

We also used $16$ payload sequences ranging from $3$ to $10$ tokens. 
These finite payload and key sets are not intended to represent the full
space of possible natural-language prompts, but rather an initial, concrete and reproducible investigation 
of the collisions question. 

For each sampled payload rank vectors $r$ and key $k$, 
we compute $w=F_{k}(r)$, and then search over $k\in\mathcal{K}_{\mathrm{adm}}$ 
and compute the fiber 
$\mathsf{ValidKeys}_{\mathcal{K}_{\mathrm{adm}}}(r,w) 
= \{k\in\mathcal{K}_{\mathrm{adm}}:F_k(r)=w\}$, 
and its size $\mu_{\mathcal{K}_{\mathrm{adm}}}(r,w)$. 
By design, we always have $\mu_{\mathcal{K}_{\mathrm{adm}}}(r_i,w_i)\geq 1$. 

\medskip
\noindent\textbf{Results.}
For all the \(960\) finite-key evaluations (\(60\) keys and \(16\) transcripts), 
the true key was contained in the candidate set in all \(16/16\) cases, confirming consistency of the finite search.  
Furthermore, no key collisions were found: every tested transcript had
candidate fiber size $1$. 
Table~\ref{tab:exp3-key-collision-results} summarizes the results.

\begin{table}[h!]
\centering
\caption{Finite key-collision search results.}
\label{tab:exp3-key-collision-results}
\begin{tabular}{l c}
\hline
\textbf{Quantity} & \textbf{Value} \\
\hline
Finite key set size \( |\mathcal{K}_{\mathrm{adm}}| \) & \(60\) \\
Tested transcripts & \(16\) \\
Finite-key evaluations & \(960\) \\
Observed collisions & \(0\) \\
Largest candidate fiber & \(1\) \\
True-key containment & \(16/16\) \\
\hline
\end{tabular}
\end{table}

The hand-crafted collision of Example~\ref{ExampleOfCollision} was not 
reproduced in this randomized finite-search run, which is consistent 
with the observation that this collision was constructed by deliberately 
exploiting subword tokenization nuances rather than discovered organically. 
This suggests that key collisions do not arise naturally under random 
key generation. 
A more extensive empirical study of collision frequency and structure, 
across larger key spaces and payload distributions, remains an 
interesting direction for future work.

\subsection{Experiment 3: Collision stability across new payloads}
\label{EmpiricalExplorations_Keycollision2}

This experiment was designed to distinguish between key pairs producing 
local (one-off) collisions and key pairs producing global (multiple) collisions, thus providing 
stronger multi-payload key ambiguity. 
Namely, if a colliding key pair
$k_a \neq k_b$ satisfying $$F_{k_a}(r_0)=F_{k_b}(r_0)$$ 
is found, then one can ask whether the equality persists on new
payload rank vectors \(r_1,\ldots,r_S\). In the present run, however, 
Experiment 2 in Section~\ref{EmpiricalExplorations_Keycollision} found no colliding key pairs. 
We therefore turn to the hand-crafted collision found in Example \ref{ExampleOfCollision}, 
which our randomized study did not encounter in Experiment 2. 
Recall that the keys $k_1 = \texttt{The quick brown fox jumps}$ and 
$k_2 = \texttt{The quick brow fox jumps}$ were shown to collide on the rank message 
$r = [1,1,1,1,1]$. We tested whether this collision persists across 
$9$ additional length-$5$ rank vectors.

\medskip
\noindent\textbf{Results.}
The collision was confirmed on $r=[1,1,1,1,1]$: both keys produced 
the stegotext \texttt{"over the lazy dog."} under the tested 
environment. Of the 9 additional rank vectors tested, only 
$r=[1,1,1,1,2]$ also produced a collision; the remaining 8 did not.
Given that the two keys diverged on most tested inputs, 
confirming the local nature of the collision, we did not 
extend the search further.
This confirms that $k_1$ and $k_2$ form a \emph{local} collision: 
they agree on specific inputs but do not induce the same map $F_k$ 
globally, which is particularly surprising given how similar these two keys are. 

We note that the exact numeric value of $w$ is environment-dependent: 
the value $w=F_{k_i}(r)=[126444,1,3736,2,4]$ reported in Example~\ref{ExampleOfCollision} 
was obtained under a specific model checkpoint, whereas the present 
run produced $w=[126166,1,4408,2,4]$. The collision phenomenon itself 
--- that $k_1$ and $k_2$ agree on $r=[1,1,1,1,1]$ --- is consistent 
across both environments.

\begin{table}[h!]
\centering
\caption{Collision stability results for the hand-crafted key pair.}
\label{tab:exp4-candidate-shrinkage}
\begin{tabular}{l c}
\hline
\textbf{Quantity} & \textbf{Value} \\
\hline
Key pair tested & $k_1, k_2$ from Example~\ref{ExampleOfCollision} \\
Rank vectors tested & \(10\) \\
Collisions confirmed & \(2\) (on $[1,1,1,1,1]$ and $[1,1,1,1,2]$) \\
Non-colliding inputs & \(8\) \\
Global collision & No \\
\hline
\end{tabular}
\end{table}

\medskip
\noindent\textbf{Interpretation.}
These results support the distinction introduced in 
Section~\ref{subsection_Key_Collisions} between local and global 
key collisions. The two keys agree on a small number of specific 
inputs, consistent with the tokenization structure that was 
deliberately exploited in Example~\ref{ExampleOfCollision}, 
but diverge on most inputs.

\subsection{Experiment 4: Non-commutativity of encoding maps}
\label{EmpiricalExplorations_Noncommutativity}

This experiment empirically studies the commutativity question raised in 
Section~\ref{Subsection_Noncommutativity}. For sampled key pairs
$(k,h)$ and payload rank vectors $r_1,\ldots,r_S$, we measure if the two compositions
\[F_k(F_h(r_j))\qquad\text{and}\qquad F_h(F_k(r_j)) \]
are equal or not. We also measure how much they differ, using the \textit{commutation} metric 
\[ \widehat{D}_{\log}(k,h) = \frac{1}{S} \sum_{j=1}^{S} d_{\log}(F_k(F_h(r_j)),F_h(F_k(r_j))), \] 
where $d_{\log}$ denotes the normalized log-rank distance, given by 
\[d_{\log}(u,v)= \frac{1}{n}\sum_{i=1}^n \frac{|\log(1+u_i)-\log(1+v_i)|}{\log(1+N)}. \] 
This measures the normalized positional difference between two rank 
vectors, with logarithmic weighting so that differences among 
high-ranked tokens contribute less than differences among 
low-ranked tokens. 

\medskip
\noindent\textbf{Results.}
No commuting key pairs were found among the 36 tested pairs.
The median commutation distance was \(0.1813\), with
empirical 5th and 95th percentiles \(0.1424\) and \(0.2237\). 
Intuitively, this means that swapping the order of two randomly chosen keys produces rank vectors that differ by roughly 18\% of the maximum possible log-rank distance, confirming that the two compositions lead to substantially different stegotexts in practice. 

\begin{table}[h!]
\centering
\caption{Non-commutativity results for key-induced maps.}
\label{tab:exp5-noncommutativity}
\begin{tabular}{l c}
\hline
\textbf{Quantity} & \textbf{Value} \\
\hline
Sampled key pairs & \(36\) \\
Payload rank vectors per pair & \(8\) \\
Commuting pairs found & \(0\) \\
Median \(\widehat{D}_{\log}\) & \(0.1813\) \\
5th percentile of \(\widehat{D}_{\log}\) & \(0.1424\) \\
95th percentile of \(\widehat{D}_{\log}\) & \(0.2237\) \\
\hline
\end{tabular}
\end{table}

Figure~\ref{fig:exp5-commutation-distance-cdf} shows the empirical CDF
of the commutation distances.
\begin{figure}[h!]
\centering
\includegraphics[width=0.45\textwidth]{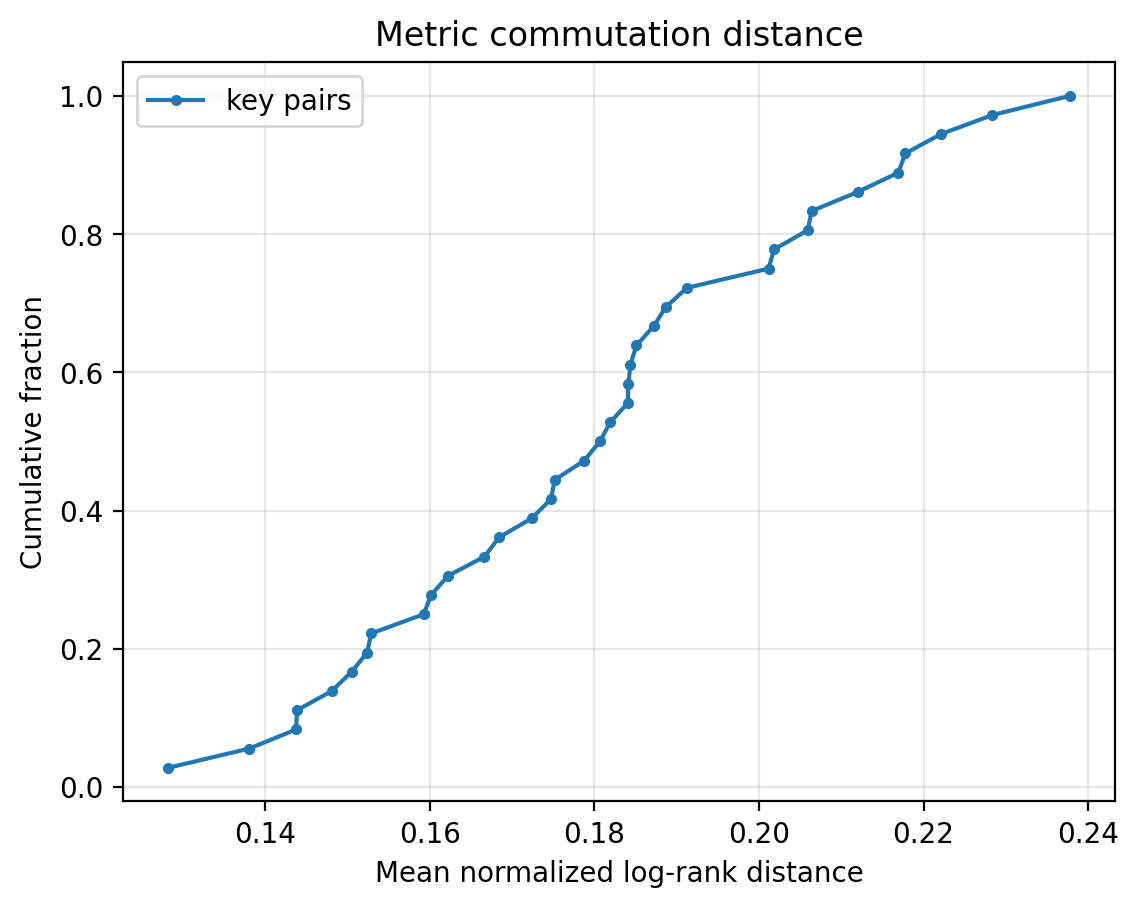}
\caption{\textbf{Metric non-commutativity of key-induced maps.}
Empirical CDF of \(\widehat{D}_{\log}(k,h)\) over sampled key pairs.
No exact sampled commuting pairs were observed; the nonzero distances
show that the order of applying \(F_k\) and \(F_h\) generally changed
the resulting rank vector in the tested sample.}
\label{fig:exp5-commutation-distance-cdf}
\end{figure}

\medskip
\noindent\textbf{Interpretation.}
These results suggest that, in practice, the encoding maps 
$E_k$ and $E_h$ do not commute for randomly sampled key pairs. 
This is a favorable property from a security perspective, as it 
indicates that the instance-generation attack described in 
Section~\ref{Subsection_Noncommutativity} is unlikely to be 
applicable in typical usage scenarios.

\subsection{Experiment 5: Robustness to token perturbations}
\label{EmpiricalExplorations_Robustness}

This experiment studies the brittleness described in Remark~\ref{rmk:error-propagation}, 
that is, whether small changes to the transmitted stegotext can disrupt exact decoding. 
For each generated stegotext $y=E_k(x)$, we construct a perturbed stegotext $\tilde{y}$ 
in one of the following ways:
\begin{enumerate}
    \item \emph{Random token substitution:} one token position is
    selected and replaced by a different random admissible token.
    \item \emph{Nearby-rank substitution:} one token is replaced by a
    token selected from a nearby rank neighborhood under the local model
    ranking.
    \item \emph{Adjacent-token transposition:} adjacent tokens are 
    swapped.
    \item \emph{Punctuation-token substitution:} one token is replaced
    by a punctuation-like token when such a replacement was available.
\end{enumerate} 

We then decode $x:=D_k(y)$ and $\tilde{x}:=D_k(\tilde{y})$ and compare them using the standard notion of 
normalized edit distance, as well as the suffix corruption fraction $\operatorname{SuffixErr}$, 
which is given by  
\[ \operatorname{SuffixErr}
= \frac{1}{n-j_{\mathrm{first}}+1} \sum_{j=j_{\mathrm{first}}}^{n} \mathbf{1}[x_j\neq \tilde{x}_j], \]
where $j_{\mathrm{first}}=\min\{j:x_j\neq \tilde{x}_j\}$ is the first mismatch position.

\medskip
\noindent\textbf{Results.}
The robustness experiment shows that the base \CARTS construction is
highly sensitive to token perturbations. Across 
\(80\) tested perturbations, decoding was corrupted in all 
cases. 
The average token edit distance between the original and corrupted 
decoded payload was $4.200$. 
Suffix corruption was total (namely, equal to $1.000$) in all perturbation types except adjacent-token transposition, 
which had average suffix corruption \(0.993\).
Table~\ref{tab:exp6-robustness-results} summarizes the perturbation 
results by perturbation type. 

\begin{table}[h!]
\centering
\small
\caption{Robustness results under stegotext perturbations.}
\label{tab:exp6-robustness-results}
\begin{tabular}
{p{0.20\textwidth}p{0.08\textwidth}p{0.08\textwidth}}
\hline
\textbf{Perturbation type} &
\textbf{Avg. edit distance} &
\textbf{Avg. suffix corruption} \\
\hline
Random token substitution & \(3.95\) & \(1.000\) \\
Nearby-rank substitution & \(3.80\) & \(1.000\) \\
Adjacent-token transposition & \(4.60\) & \(0.993\) \\
Punctuation-token substitution & \(4.45\) & \(1.000\) \\
\hline
\end{tabular}
\end{table}

\medskip
\noindent\textbf{Interpretation.}
The results shown in Table~\ref{tab:exp6-robustness-results} 
and Figure~\ref{fig:exp6-robustness-edit-distance} 
confirm that \CARTS provides exact invertibility 
under matched conditions, but is brittle to channel noise.  
This is consistent with the autoregressive error-propagation structure 
described in Remark~\ref{rmk:error-propagation}. Robust variants 
would require additional error-correction mechanisms, and are 
left for future work.

\begin{figure}[h]
\centering
\includegraphics[width=0.45\textwidth]{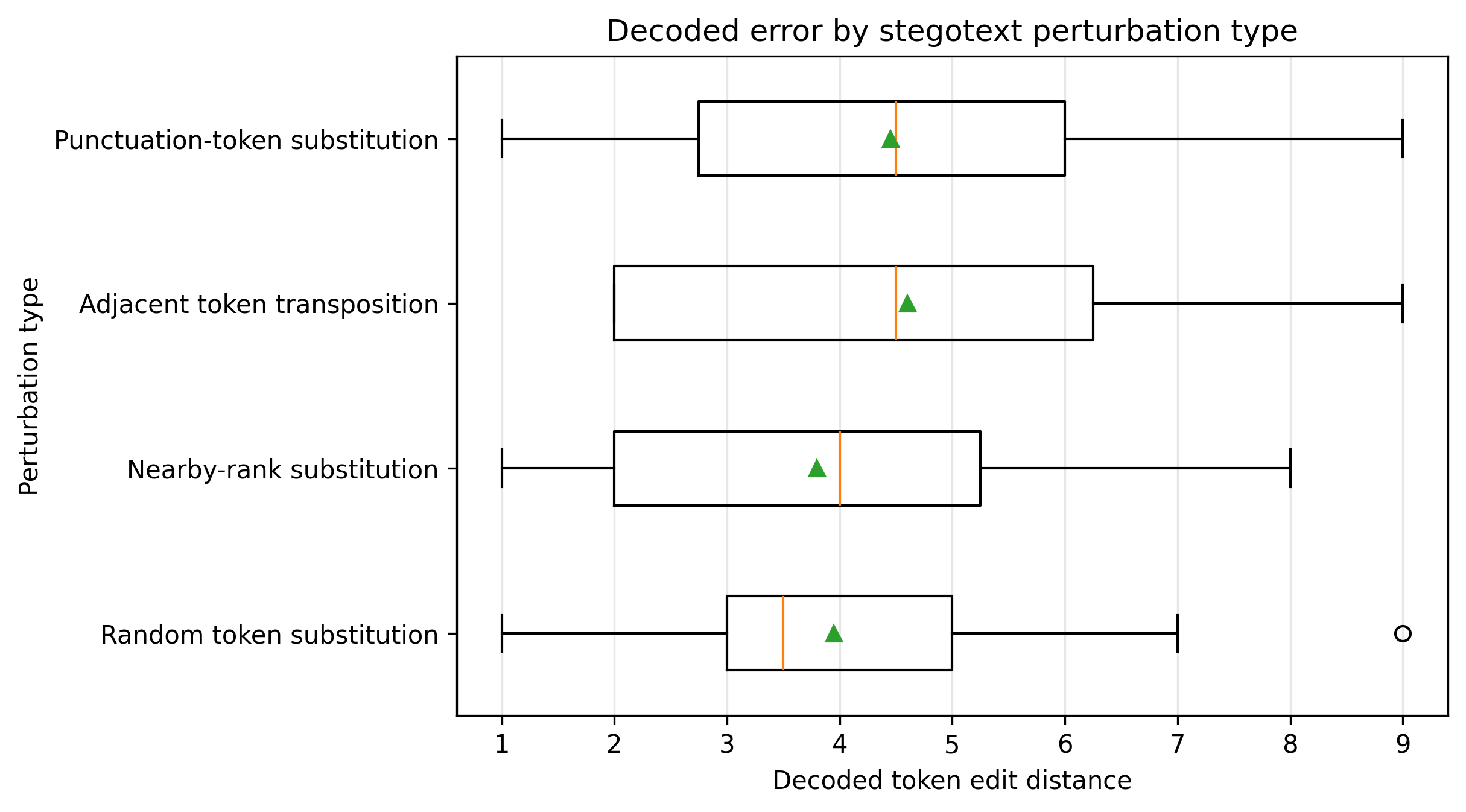}
\caption{Decoded token edit distance grouped by perturbation type.
All perturbation types consistently disrupted decoding.}
\label{fig:exp6-robustness-edit-distance}
\end{figure}

\subsection{Summary of empirical investigations}
\label{EmpiricalExplorations_MinimumPackage}

The five experiments give a consistent empirical picture of the base 
\CARTS protocol. Exact token-level recovery succeeded in all 
\(40/40\) tested payload-key pairs, confirming that the implementation 
correctly realizes the deterministic rank-transcoding mechanism. 
The finite key-collision search found no collisions among \(960\) 
evaluations, and the hand-crafted collision of 
Example~\ref{ExampleOfCollision} was confirmed to be local, 
persisting on only 2 out of 10 tested rank vectors. 
No commuting key pairs were found among the 36 tested pairs, 
suggesting that an instance-generation attack 
is unlikely to apply in practice. 
Finally, all 80 tested token perturbations corrupted 
decoding, confirming the autoregressive error-propagation structure 
described in Remark~\ref{rmk:error-propagation}.

Together, these results validate the formal framework introduced in 
Sections~\ref{sec:carts-protocol} and~\ref{sec:security}: \CARTS 
provides exact, full-token-rate, deterministic keyed rank-transcoding, 
and the empirical evidence suggests that the construction is resistant 
to the specific attack vectors studied here.

\section{Conclusion} 

This paper introduced \CARTS as a formal framework for keyed 
text-to-text rank-transcoding steganography, formalizing the Calgacus 
construction \cite{norelli2025llms} and establishing its exact 
correctness under deterministic model assumptions. Within this 
framework, we defined the computational problems naturally associated 
with the construction --- context search, key collisions, message 
equivocation, and non-commutativity of key-induced maps --- and studied 
their theoretical properties and empirical behavior, 
providing the first rigorous treatment of the security landscape of 
this promising but previously unstudied class of protocols. 
The experiments confirmed exact encoding-decoding recovery, found no key collisions under 
random key generation, established that the hand-crafted collision of 
Example~\ref{ExampleOfCollision} is local, and found no commuting key 
pairs, suggesting that the construction is resistant to the specific 
attack vectors studied here.

A primary goal of this work is to open up this area for systematic investigation. 
On the theoretical side, the hardness of the context search problem, 
the asymptotic structure of collision fibers, 
the conditions under which commuting key pairs exist, 
and the design of more general CARTS transforms 
are all concrete open problems that follow naturally from the framework. 
On the empirical side, 
larger key spaces, more diverse payload distributions, and semantic 
quality metrics would give a more complete picture of the construction's behavior. 

We hope the formal framework introduced here provides a solid foundation for the constructive use of 
autoregressive language models in steganography and privacy-preserving 
communication in general.

\bibliographystyle{ACM-Reference-Format}
\bibliography{sample-base}

\end{document}